\documentclass{llncs}
\usepackage{makeidx}

\usepackage{amssymb}
\usepackage{amsmath}
\usepackage{color}
\usepackage{colortbl}
\usepackage{graphicx}

\usepackage[bookmarksnumbered=true,
            bookmarksopen=true,
            colorlinks=true,
            citecolor=red,
            linkcolor=blue,
            anchorcolor=red,
            urlcolor=blue
            ]{hyperref}

\newcommand{\eps}{\epsilon}
\newcommand{\intd}{{\rm d}}
\newcommand{\E}{\mathbb{E}}
\newcommand{\R}{{\mathbb{R}}}

\DeclareMathOperator{\tr}{tr}

\newcommand{\hnnote}[1]{}

\usepackage{prettyref}
\newrefformat{eq}{(\ref{#1})}
\newrefformat{sec}{Section~\ref{#1}}
\newrefformat{fig}{Fig.~\ref{#1}}
\newrefformat{tab}{Table~\ref{#1}}
\newrefformat{rmk}{Remark~\ref{#1}}
\newrefformat{def}{Definition~\ref{#1}}
\newrefformat{cor}{Corollary~\ref{#1}}
\newrefformat{lmm}{Lemma~\ref{#1}}
\newrefformat{prop}{Proposition~\ref{#1}}
\newrefformat{app}{Appendix~\ref{#1}}
\newcommand{\Th}{{^{\rm th}}}
\newcommand{\diff}{{\rm d}}
\newcommand{\calE}{\mathcal{E}}

\begin{document}
\title{Tight Lower Bound for Linear Sketches of Moments}
\titlerunning{Sketches of Moments}
\author{Alexandr Andoni\inst{1} \and Huy L. Nguy$\tilde{\mbox{\^{e}}}$n\inst{2} \and Yury Polyanskiy\inst{3} \and Yihong Wu\inst{4}}
\institute{Microsoft Research SVC, \email{andoni@microsoft.com}
\and
Princeton U, \email{hlnguyen@princeton.edu}
\and
MIT, \email{yp@mit.edu}
\and
UIUC, \email{yihongwu@illinois.edu}
}
\authorrunning{Alexandr Andoni et al.}
\maketitle

\begin{abstract}
The problem of estimating frequency moments of a data stream has
attracted a lot of attention since the onset of streaming
algorithms~\cite{AMS}. While the space complexity for approximately
computing the $p^{\rm th}$ moment, for $p\in(0,2]$ has been settled
  \cite{KNW-exact10}, for $p>2$ the exact complexity remains open.
  For $p>2$ the current best algorithm uses $O(n^{1-2/p}\log n)$ words
  of space \cite{AKO-precision,BO10}, whereas the lower bound is of
  $\Omega(n^{1-2/p})$ \cite{BJKS}.

In this paper, we show a tight lower bound of $\Omega(n^{1-2/p}\log
n)$ words for the class of algorithms based on linear sketches, which
store only a sketch $Ax$ of input vector $x$ and some (possibly
randomized) matrix $A$. We note that all known algorithms for this
problem are linear sketches.
\end{abstract}

\section{Introduction}

One of the classical problems in the streaming literature is that of
computing the $p$-frequency moments (or $p$-norm) \cite{AMS}. In
particular, the question is to compute the norm $\|x\|_p$ of a vector
$x\in \R^n$, up to $1+\eps$ approximation, in the streaming model
using low space. Here, we assume the most general model of streaming,
where one sees updates to $x$ of the form $(i,\delta_i)$ which means
to add a quantity $\delta_i\in \R$ to the coordinate $i$ of
$x$.\footnote{For simplicity of presentation, we assume that
  $\delta_i\in \{-n^{O(1)}, \ldots, n^{O(1)}\}$, although more refined
  bounds can be stated otherwise. Note that in this case, a ``word''
  (or measurement in the case of linear sketch --- see definition
  below) is usually $O(\log n)$ bits.} In this setting, linear
estimators, which store $Ax$ for a matrix $A$, are particularly useful
as such an update can be easily processed due to the equality
$A(x+\delta_i e_i)=Ax+A(\delta_ie_i)$.

The frequency moments problem is among the problems that received the
most attention in the streaming literature. For example, the space
complexity for $p\le 2$ has been fully understood. Specifically, for
$p=2$, the foundational paper of \cite{AMS} showed that $O_\eps(1)$
words (linear measurements) suffice to approximate the Euclidean
norm\footnote{The exact bound is $O(1/\eps^2)$ words; since in this
  paper we concentrate on the case of $\eps=\Omega(1)$ only, we drop
  dependence on $\eps$.}. Later work showed how to achieve the same
space for all $p\in(0,2)$ norms
\cite{I00b,Li-estimators08,KNW-exact10}. This upper bound has a
matching lower bound \cite{AMS,IW03,Y-phd,W-optimalSpace}. Further
research focused on other aspects, such as algorithms with improved
update time (time to process an update $(i,\delta_i)$)
\cite{NW-fastsketches10,KNW-exact10,Li-estimators08,GC-moments07,KNPW10-otimalUpdate}.

In constrast, when $p>2$, the exact space complexity still remains
open. After a line of research on both upper bounds
\cite{AMS,IW05,Ganguly-Fk06,MW-l1sampling},
\cite{AKO-precision,BO10,Ganguly-polyEstimators} and lower bounds
\cite{AMS,ChakrabartiKS03,BJKS,JST-sampling,PW12}, we presently know
that the best space upper bound is of $O(n^{1-2/p}\log n)$ words, and
the lower bound is $\Omega(n^{1-2/p})$ bits (or linear
measurements). (Very recently also, in a {\em restricted} streaming
model --- when $\delta_i=1$ --- \cite{BO12-pickanddrop} achieves an
improved upper bound of nearly $O(n^{1-2/p})$ words.) In fact, since
for $p=\infty$ the right bound is $O(n)$ (without the log factor), it
may be tempting to assume that there the right upper bound should be
$O(n^{1-2/p})$ in the general case as well.

In this work, we prove a tight lower bound of $\Omega(n^{1-2/p}\log
n)$ for the case of {\em linear estimator}. A linear estimator uses a
distribution over $m\times n$ matrices $A$ such that with high
probability over the choice of $A$, it is possible to calculate the
$p^{\rm th}$ moment $\|x\|_p$ from the sketch $Ax$. The parameter $m$,
the number of words used by the algorithm, is also called the number
of measurements of the algorithm. Our new lower bound is of
$\Omega(n^{1-2/p}\log n)$ measurements/words, which matches the upper
bound from \cite{AKO-precision,BO10}. We stress that essentially all
known algorithms in the general streaming model are in fact linear
estimators.

\begin{theorem}
\label{thm:fkLowerBound}
Fix $p \in (2,\infty)$. Any linear sketching algorithm for approximating the $p\Th$ moment
of a vector $x\in \R^n$ up to a multiplicative factor $2$ with probability
$99/100$ requires $\Omega(n^{1-2/p}\log n)$ measurements.

In other words, for any $p\in(2,\infty)$ there is a constant $C_p$ such that for any distribution on $m \times n$ matrices $A$ with $m < C_p n^{1-2/p}\log n$ and any function $f: \mathbb{R}^{m\times n}
\times \mathbb{R}^m \to \mathbb{R}_+$ we have

\begin{equation}
	\inf_{x\in \mathbb{R}^n} \Pr\left( {1\over 2} \|x\|_p \le f(A, Ax) \le 2\|x\|_p
\right) \leq {99\over 100} \, .
	\label{eq:minimax}
\end{equation}
\end{theorem}

The proof uses similar hard distributions as in some of the previous
work, namely all coordinates of an input vector $x$ have random small
values except for possibly one location. To succeed on these
distributions, the algorithm has to distinguish between a mixture of
Gaussian distributions and a pure Gaussian distribution. Analyzing the
optimal probability of success directly seems too difficult. Instead,
we use the $\chi^2$-divergence to bound the success probability, which
turns out to be much more amenable to analysis.

From a statistical perspective, the problem of linear sketches of
moments can be recast as a minimax statistical estimation problem
where one observes the pair $(Ax,A)$ and produces an estimate of
$\|x\|_p$. More specifically, this is a functional estimation problem,
where the goal is to estimation some functional (in this case, the
$p\Th$ moment) of the parameter $x$ instead of estimating $x$
directly.  Under this decision-theoretic framework, our argument can
be understood as Le Cam's two-point method for deriving minimax lower
bounds \cite{LeCam86}. The idea is to use a binary hypotheses testing
argument where two priors (distributions of $x$) are constructed, such
that 1) the $p\Th$ moment of $x$ differs by a constant factor under
the respective prior; 2) the resulting distributions of the sketches
$Ax$ are indistinguishable. Consequently there exists no moment
estimator which can achieve constant relative error. This approach is
also known as the method of fuzzy hypotheses \cite[Section
  2.7.4]{Tsybakov09}. See also \cite{BL96,IS03,Low10,CL11} for the
method of using $\chi^2$-divergence in minimax lower bound.

We remark that our proof does not give a lower bound as a function of
$\eps$ (but \cite{W13-personal} independently reports progress on this
front).

\subsection{Preliminaries}

We use the following definition of divergences.

\begin{definition}
Let $P$ and $Q$ be probability measures.
The {\em $\chi^2$-divergence} from $P$ to $Q$ is
	\begin{eqnarray*}
			\chi^2(P || Q) &\triangleq& \int \left(\frac{\intd P}{\intd Q} -
1\right)^2 \intd Q\\
				&=& \int \left(\frac{\intd P}{\intd Q}\right)^2 \intd Q -
1
	\end{eqnarray*}
The \emph{total variation distance} between $P$ and $Q$ is
\begin{equation}
V(P,Q) \triangleq 	\sup_{A} |P(A) - Q(A)| = \frac{1}{2} \int |\diff P - \diff Q|
	\label{eq:TV}
\end{equation}
\end{definition}

The operational meaning of the total variation distance is as follows: 
Denote the optimal sum of Type-I and Type-II error probabilities of the binary hypotheses testing problem $H_0: X \sim P$ versus $H_1: X \sim Q$ by
\begin{equation}
	\calE(P,Q) \triangleq \inf_{A} \{P(A) + Q(A^{\rm c})\},
	\label{eq:perr}
\end{equation}
where the infimum is over all measurable sets $A$ and the corresponding test is to declare $H_1$ if and only if $X \in A$. 
Then
\begin{equation}
\calE(P,Q) = 1 - V(P,Q).	
	\label{eq:ETV}
\end{equation}

The total variation and the $\chi^2$-divergence are related by the following inequality 
\cite[Section 2.4.1]{Tsybakov09}:
\begin{equation}
	2 V^2(P,Q) \leq \log(1+\chi^2(P||Q))
	\label{eq:tvchi2}
\end{equation}
Therefore, in order to establish that two hypotheses cannot be distinguished with vanishing error probability, it suffices to show that the $\chi^2$-divergence is bounded. 

One additional fact about $V$ and $\chi^2$ is the data-processing property~\cite{Csiszar67}: If
a measurable function $f:A\to B$ carries probability measure $P$ on $A$  to $P'$ on $B$, and carries $Q$ to $Q'$ then
\begin{equation}\label{eq:tvdproc}
	 V(P, Q) \ge V(P', Q')\,.
\end{equation}

\section{Lower Bound Proof}

In this section we prove Theorem~\ref{thm:fkLowerBound} for arbitrary fixed measurement matrix $A$.
Indeed, by Yao's minimax principle, we only need to demonstrate an input
distribution and show that any deterministic algorithm succeeding on
this distribution with probability 99/100 must use
$\Omega(n^{1-2/p}\log n)$ measurements.

Fix $p \in (2,\infty)$. Let $A \in \R^{m \times n}$ be a fixed matrix which is used to produce the linear
sketch, where $m<C_p n^{1-2/p}\log n$ is the number of measurements and $C_p$ is to be specified.
Next, we construct distributions $D_1$ and $D_2$ for $x$ to fulfill the following properties: 
\begin{enumerate}
	\item $\|x\|_p \leq C n^{1/p}$  on the entire support of $D_1$, and $\|x\|_p \geq 4 C n^{1/p}$ on the entire support of $D_2$, for some appropriately chosen constant $C$.
	\item Let $E_1$ and $E_2$ denote the distribution of $Ax$ when $x$ is drawn from $D_1$ and $D_2$ respectively. Then $V(E_1,E_2) \leq 98/100$.
\end{enumerate}  
The above claims immediately imply the desired \prettyref{eq:minimax} via the relationship between statistical tests and estimators.
 To see this, note that any moment estimator $f$ induces a test for distinguishing $E_1$ versus
$E_2$: declare $D_2$ if and only if $\frac{f(A,Ax)}{2C n^{1/p}} \geq 1$. 
In other words,
\begin{align}
\lefteqn{\Pr_{x \sim \frac{1}{2}(D_1+D_2)}\left( {1\over 2} \|x\|_p \le f(A, Ax) \le 2\|x\|_p
\right)} \nonumber\\
	& \le {1\over 2} \Pr_{x \sim D_2}\left(f(A, Ax) > 2Cn^{1/p} \right) + 
		{1\over 2} \Pr_{x \sim D_1}\left(f(A, Ax) \le 2Cn^{1/p} \right) \\
	& \le {1\over 2} (1+V(E_1, E_2)) \le {99\over100}\,,
\end{align}
where the last line follows from the characterization of the total variation in \prettyref{eq:TV}.

The idea for constructing the desired pair of distributions is to use the Gaussian distribution and its sparse perturbation. Since the moment of a Gaussian random vector takes values on the entire $\R_+$, we need to further truncate by taking its conditioned version.
To this end, let $y\sim N(0,I_n)$ be a standard normal random vector and $t$ a random index uniformly distributed on $\{1,\ldots,n\}$ and independently of $y$. Let $\{e_1,\ldots,e_n\}$ denote the standard basis of $\R^n$. Let $\bar D_1$ and $\bar D_2$ be input distributions defined as follows:
Under the distribution $\bar D_1$, we let the input vector
$x$ equal to $y$. Under the distribution $\bar D_2$, we add a one-sparse perturbation by setting $x = y + C_1 n^{1/p} e_t$ with an appropriately chosen constant $C_1$.
Now we set $D_1$ to be $\bar D_1$ conditioned on the event $E=\{z: \|z\|_p \leq C n^{1/p}\}$, i.e., $D_1(\cdot) = \frac{\bar D_1(\cdot \cap E)}{\bar D_1(E)}$, 
 and set $D_2$ to be $\bar D_2$ conditioned on the event $F=\{z: \|z\|_p \geq 4 C n^{1/p}\}$. By the triangle inequality,
\begin{align}
V(E_1,E_2)
\leq & ~ V(\bar E_1, \bar E_2) + V(\bar E_1, E_1) + V(\bar E_2, E_2) \nonumber \\
\leq & ~ V(\bar E_1, \bar E_2) + V(\bar D_1, D_1) + V(\bar D_2, D_2) \nonumber \\
= & ~ V(\bar E_1, \bar E_2) + \Pr_{x \sim \bar D_1}(\|x\|_p \geq C n^{1/p}) + \Pr_{x \sim \bar
D_2}(\|x\|_p \leq 4 C n^{1/p})  \label{eq:tv3},
\end{align}
where the second inequality follows from the data-processing inequality~\eqref{eq:tvdproc} (applied
to the mapping $x \mapsto Ax$). It remains to bound the three terms in \prettyref{eq:tv3}.

First observe that for any $i$, $\E[|y_i|^p] = t_p$ where
$t_p=2^{p/2}\Gamma(\tfrac{p+1}{2})\pi^{-1/2}$. Thus, $\E[\|y\|_p^p] =
nt_p$. By Markov inequality, $\|y\|_p^p \ge 100 nt_p$ holds with probability at most $1/100$. Now, if we set 
\begin{equation}
C_1 = 4\cdot
(100t_p)^{1/p}+10,	
	\label{eq:C1}
\end{equation}
 we have $(y_t + C_1 n^{1/p})^p > 4^p\cdot 100 nt_p $
with probability at least $99/100$, and hence the third term in~\eqref{eq:tv3} is also smaller than
$1\over 100$. It remains to show that $V(\bar E_1, \bar E_2) \leq 96/100$.


Without loss of generality, we assume that the rows of $A$ are orthonormal
since we can always change the basis of $A$ after taking the
measurements. Let $\eps$ be a constant smaller than $1-2/p$. Assume that $m < \tfrac{\eps}{100C_1^2}
\cdot n^{1-2/p}\log n$. Let $A_i$ denote the $i\Th$ column of $A$. Let $S$ be the set of
indices $i$ such that $\|A_i\|_2 \le 10\sqrt{m/n} \leq
n^{-1/p}\sqrt{\eps\log n}/C_1$. Let $\bar{S}$ be the complement of $S$. Since $\sum_{i=1}^n \|A_i\|_2^2 = m$, we have
$|\bar{S}| \leq n/100$. Let $s$ be uniformly distributed on $S$ and $\tilde{E}_2$ the distribution of $y + C_1 n^{1/p} e_s$. By the convexity of $(P,Q) \mapsto V(P,Q)$ and the fact that $V(P,Q)\leq 1$, we have 
$V(\bar E_1,\bar E_2) \leq V(\bar E_1, \tilde E_2)	+ \frac{|\bar S|}{n} \leq V(\bar E_1, \tilde E_2) + 1/100$. In view of \prettyref{eq:tvchi2}, it suffices to show that 
\begin{equation}
	\chi^2(\tilde E_2\|\bar E_1) \leq c
	\label{eq:chi2c}
\end{equation}
for some sufficiently small constant $c$. To this end, we first prove a useful fact about the measurement matrix $A$.


%
%
%

\begin{lemma}\label{lem:eqn1}
For any matrix $A$ with $m< \tfrac{\eps}{100C_1^2}
\cdot n^{1-2/p}\log n$ orthonormal rows, denote by $S$
the set of column indices $i$ such that $\|A_i\|_2 \le 10\sqrt{m/n}$. Then
$$|S|^{-2}\sum_{i,j\in S} e^{C_1^2 n^{2/p}\langle A_i, A_j\rangle} \le 1.03C_1^4(n^{-2+4/p+\eps}m+n^{2/p-1}\sqrt{m}) + 1$$
\end{lemma}
\begin{proof}
Because $AA^T = I_m$, we have
$$\sum_{i,j\in [n]} \langle A_i,A_j\rangle^2 = \sum_{i,j\in [n]} (A^T A)_{ij}^2 = \|A^TA\|_F^2 = \tr(A^TAA^TA) = \tr(A^TA) = \|A\|_F^2 = m.$$

We consider the following relaxation: let $x_{1}, \dots, x_{|S|^2} \ge 0$ where $\sum_i x_i^2 \le C_1^4 n^{4/p}\cdot m$ and $x_i \le \eps\log n$. We now upper bound $|S|^{-2}\sum_{i=1}^{|S|^2} e^{x_i}$. We have
\begin{align*}
|S|^{-2}\sum_{i=1}^{|S|^2}e^{x_i} &= |S|^{-2}\sum_{i=1}^{|S|^2} \left(1+x_i + \sum_{j\ge 2} \frac{x_i^j}{j!}\right)\\
&\le 1+|S|^{-2}\sum_{i=1}^{|S|^2} x_i + |S|^{-2}\sum_{i=1}^{|S|^2} x_i^2\sum_{j\ge 2} \frac{(\max_{i\in[n^2]} x_i)^{j-2}}{j!}\\
&\le 1+|S|^{-2}\sqrt{|S|^2\sum_i x_i^2}+|S|^{-2}(C_1^4mn^{4/p})\left(\frac{e^{\eps \log n}}{(\eps\log n)^2}\right)\\
&\le 1+1.03C_1^2 \sqrt{m}n^{2/p-1} + 1.03 C_1^4n^{-2+4/p+\eps}m.
\end{align*}
 The last inequality uses the fact that $99n/100\le|S|\le n$.
Applying the above upper bound to $x_{(i-1)|S|+j} =
C_1^2n^{2/p}|\langle A_i, A_j\rangle|\le
C_1^2n^{2/p}\|A_i\|\cdot\|A_j\|\le \eps\log n$, we conclude the lemma.
\end{proof}

We also need the following lemma \cite[p. 97]{IS03} which gives a formula for the $\chi^2$-divergence from a Gaussian location mixture to a standard Gaussian distribution:
\begin{lemma}
Let $P$ be a distribution on $\mathbb{R}^m$. Then
\[
\chi^2( N(0,I_m) * P \, || \, N(0,I_m)) = \mathbb{E}[\exp(\langle X, X'\rangle)] -1 \, ,
\]
where $X$ and $X'$ are independently drawn from $P$.
	\label{lmm:chi2}
\end{lemma}

We now proceed to proving an upper bound on the $\chi^2$-divergence
between $\bar E_1$ and $\tilde E_2$.

\begin{lemma}
$$\chi^2(\tilde E_2\|\bar E_1) \le 1.03C_1^4(n^{-2+4/p+\eps}m + n^{2/p-1}\sqrt{m})$$
\end{lemma}
\begin{proof}
Let $p_i=1/|S|~\forall i\in S$ be the probability $t=i$. 
Recall that $s$ is the random index uniform on the set $S= \{i \in [n]: \|A_i\|_2 \le 10\sqrt{m/n} \}$.
Note that $Ay \sim N(0,AA^T)$. Since 
$AA^T=I_m$, we have $\bar E_1 = N(0,I_m)$. Therefore $A(y+C_1 n^{1/p}) \sim \tilde E_2 = \frac{1}{|S|} \sum_{i \in S} N(A_i,I_m)$, a Gaussian location mixture.

Applying Lemma \ref{lmm:chi2} and then Lemma \ref{lem:eqn1}, we have
\begin{align}
\nonumber \chi^2(\tilde E_2\|\bar E_1)
&= \sum_{i,j\in S} p_i p_j e^{C_1^2 n^{2/p}\langle A_i, A_j\rangle}-1\\
\nonumber
&\le 1.03C_1^4(n^{-2+4/p+\eps}m + n^{2/p-1}\sqrt{m}).
\end{align}
\end{proof}

Finally, to finish the lower bound proof, since $\eps < 1-2/p$ we have $n^{-2+4/p+\eps}m + n^{2/p-1}\sqrt{m} =
o(1)$, implying~\eqref{eq:chi2c} for all sufficiently large $n$ and completing the proof of $V(E_1,E_2) \le {98/100}$.

\section{Discussions}
\label{sec:discussion}

While Theorem~\ref{thm:fkLowerBound} is stated only for constant $p$, the proof also gives lower bounds for $p$ depending on $n$.
At one extreme, the proof recovers the known lower bound for approximating the $\ell_{\infty}$-norm of $\Omega(n)$. Notice that the ratio between the $\ell_{(\ln n)/\varepsilon}$-norm and the $\ell_{\infty}$-norm of any vector is bounded by $e^{\varepsilon}$ so it suffices to consider $p=(\ln n)/\varepsilon$ with a sufficiently small constant $\varepsilon$. Applying the Stirling approximation to the crude value of $C_1$ in the proof, we get $C_1 = \Theta(\sqrt{p})$. Thus, the lower bound we obtain is $\Omega(n^{1-2/p}(\log n) / C_1^2) = \Omega(n)$.

At the other extreme, when $p\rightarrow 2$, the proof also gives super constant lower bounds up to $p=2+\Theta(\log\log n/\log n)$. Notice that $\eps$ can be set to $1-2/p-\Theta(\log\log n/\log n)$ instead of a positive constant strictly smaller than $1-2/p$. For this value of $p$, the proof gives a ${\rm polylog}(n)$ lower bound. We leave it as an open question to obtain tight bounds for $p=2+o(1)$.

\subsubsection*{Acknowledgments.}

HN was supported by NSF CCF 0832797, and a Gordon Wu Fellowship.
YP's work was supported by the Center
for Science of Information (CSoI), an NSF Science and Technology
Center, under grant agreement CCF-0939370.


\def\cprime{$'$} \def\cprime{$'$}

\end{document}